\def\draft{0}  
    \def\ITCS{0} 

    \documentclass[11pt]{article}
    \usepackage[letterpaper,hmargin=1in,vmargin=1in]{geometry}
    \ifnum\ITCS=0
    \fi
    \usepackage{verbatim,sgame}
    
    \ifnum\ITCS=1
\usepackage[utf8]{inputenc}
\usepackage{amsmath,amssymb,amsthm}
\usepackage{xcolor}
\usepackage[hidelinks]{hyperref}
\usepackage{bm}
\usepackage{nicefrac}
\usepackage{paralist}
\usepackage{typearea}
\typearea{16}
\fi
    
    \usepackage{natbib}
    
    \usepackage{tikz}
    \usetikzlibrary{positioning,chains,fit,shapes,calc}
    \usepackage{float,subcaption}
    \floatstyle{boxed}
    \restylefloat{figure, graphicx}
    \usepackage[hidelinks]{hyperref}

    \ifnum\ITCS=1
    \allowdisplaybreaks
\fi

    \usepackage{amsfonts,url,ifthen,epsfig,amsmath,amssymb,color,framed,float,mathrsfs}
    \usepackage{algorithm}
    \usepackage[noend]{algpseudocode}
    \usepackage{etoolbox}
    
    \makeatletter
    \let\original@footnotemark\footnotemark
    \newcommand{\align@footnotemark}{%
      \ifmeasuring@
        \chardef\@tempfn=\value{footnote}%
        \original@footnotemark
        \setcounter{footnote}{\@tempfn}%
      \else
        \iffirstchoice@
          \original@footnotemark
        \fi
      \fi}
    \pretocmd{\start@align}{\let\footnotemark\align@footnotemark}{}{}
    \makeatother
    
    \hypersetup{colorlinks=false}
    
    \ifnum\draft=1
    \newcommand{\Rnote}[1]{\begin{framed}\noindent \textcolor{red}{{#1}}\end{framed}} 
    \else
    \newcommand{\Rnote}[1]{}
    \fi

    \newcommand{\remove}[1]{}
    \newtheorem{theorem}{Theorem}
    \newtheorem{example}{Example}

    \newtheorem{lemma}{Lemma}
    \newtheorem{corollary}{Corollary}
    \newtheorem{obs}{Observation}
    \newtheorem{definition}{Definition}
    \newtheorem{proposition}{Proposition}
    \newtheorem{remk}[theorem]{Remark}

    \def\FullBox{\hbox{\vrule width 8pt height 8pt depth 0pt}}
    
    \def\qed{\ifmmode\qquad\FullBox\else{\unskip\nobreak\hfil
    \penalty50\hskip1em\null\nobreak\hfil\FullBox
    \parfillskip=0pt\finalhyphendemerits=0\endgraf}\fi}
    
    \def\qedsketch{\ifmmode\Box\else{\unskip\nobreak\hfil
    \penalty50\hskip1em\null\nobreak\hfil$\Box$
    \parfillskip=0pt\finalhyphendemerits=0\endgraf}\fi}
    
 \ifnum\ITCS=0    
    \newenvironment{proof}{\begin{trivlist} \item {\bf Proof:~~}}
      {\qed\end{trivlist}}
      \fi

    \newenvironment{proofof}[1]{\begin{trivlist} \item {\bf Proof of #1:~~}}
      {\qed\end{trivlist}}

    \newcommand{\beq}{\begin{equation}}
    \newcommand{\eeq}{\end{equation}}
    \newcommand{\be}{\begin{enumerate}}
    \newcommand{\ee}{\end{enumerate}}
    \newcommand{\bi}{\begin{itemize}}
    \newcommand{\ei}{\end{itemize}}
    \newcommand{\bd}{\begin{description}}
    \newcommand{\ed}{\end{description}}
    \newcommand{\bc}{\begin{center}}
    \newcommand{\ec}{\end{center}}
    \newcommand{\bthm}{\begin{theorem}}
    \newcommand{\ethm}{\end{theorem}}
    \newcommand{\bdefi}{\begin{definition}}
    \newcommand{\edefi}{\end{definition}}
    \newcommand{\bcor}{\begin{corollary}}
    \newcommand{\ecor}{\end{corollary}}
    \newcommand{\blem}{\begin{lemma}}
    \newcommand{\elem}{\end{lemma}}
    \newcommand{\bexa}{\begin{example}}
    \newcommand{\eexa}{\end{example}}
    \newcommand{\bprop}{\begin{proposition}}
    \newcommand{\eprop}{\end{proposition}}

   \newcommand{\added}[1]{}

    \ifnum\draft=1
    \newcommand{\ronote}[1]{\begin{framed}\noindent \textcolor{red}{{Ronen's note: #1}}\end{framed}} 
    \newcommand{\ranote}[1]{\begin{framed}\noindent \textcolor{red}{{Rann's note: #1}}\end{framed}} 
    \else
    \newcommand{\ranote}[1]{}
    \newcommand{\ronote}[1]{}
    \fi

    \def\real{\hbox{\rm\setbox1=\hbox{I}\copy1\kern-.45\wd1 R}}
    \def\neal{\hbox{\rm\setbox1=\hbox{I}\copy1\kern-.45\wd1 N}}

    \newcommand{\abs}[1]{\left|{#1}\right|}
    \newcommand{\eps}{\varepsilon}

\begin{document}
    
    \definecolor{myblue}{RGB}{80,80,160}
    \definecolor{mygreen}{RGB}{80,160,80}
    
    \title{Pareto-Improving Data-Sharing}
 \ifnum\ITCS=1
 \author{}
 \else   
    \author{Ronen Gradwohl\thanks{Department of Economics and Business Administration, Ariel University. Email: \texttt{roneng@ariel.ac.il}.
    Gradwohl gratefully acknowledges the support of National Science Foundation award number 1718670.
    }
    \and
    Moshe Tennenholtz\thanks{Faculty of Industrial Engineering and Management, The Technion -- Israel Institute of
    Technology. Email: \texttt{moshet@ie.technion.ac.il.} The work by Moshe Tennenholtz was supported by funding from the European
Research Council (ERC) under the European Union's Horizon 2020
research and innovation programme (grant number 740435).}}
    \fi
    \date{}
    
	\maketitle

       \begin{abstract}
       
We study the effects of data sharing between firms on prices, profits, and consumer welfare. Although
indiscriminate sharing of consumer data decreases firm profits due to the subsequent increase in competition,
selective sharing can be beneficial. We show that there are data-sharing mechanisms that are strictly Pareto-improving, simultaneously increasing firm profits and consumer welfare. Within the class of Pareto-improving mechanisms, we identify one that maximizes firm profits and one that maximizes consumer welfare.

       \end{abstract}

%
%
    
    \renewcommand{\thefootnote}{\arabic{footnote}}
    \setcounter{footnote}{0}

\section{Introduction}
Data-driven innovation promises to transform the economic landscape and bring tremendous benefits to individuals. One of the key ingredients to such innovation consists of data-sharing between firms. The importance of such sharing is underscored by the European Commission's Strategy for Data, which describes the Commission's plan to ``invest in a High Impact Project on...infrastructures, data-sharing tools, architectures and governance mechanisms for thriving data-sharing and Artificial Intelligence ecosystems.'' Private industry is also rapidly developing such tools,  including platforms such as Google Merchant and Azure Data Share that facilitate data sharing between firms. Finally, the growing area of federated machine learning focuses on designing algorithms that enable firms to share data used by their respective predictive analytics tools \citep{yang2019federated,rasouli2021data}.

Despite the investment in and proliferation of data-sharing tools, the full realization of a data-driven transformation must overcome some hurdles. One of the major challenges is spelled out by the European Commission as follows: ``In spite of the economic potential, data sharing between companies has not taken off at sufficient scale. This is due to a lack of economic incentives (including the fear of losing a competitive edge),'' \citep{EC2020}. \added{Another substantial challenge consists of individuals' concerns regarding the privacy of their data, which may be exacerbated when that data is shared between firms.}

%

In this paper we formally study firms'  incentives for data sharing as well as individuals' gains from it, with a focus on the possible existence of data-sharing mechanisms that both benefit individuals and incentivize firms to participate.  We show that firms' hesitancy may be justified, in that indiscriminate data-sharing can be mutually harmful.\footnote{This is in line with a main insight of the literature on competitive price-discrimination---see the literature review below.} However, we also show that data sharing need not be a zero sum endeavor---benefiting consumers at the expense of firms, or vice versa---and that more-carefully designed mechanisms for partial data-sharing can simultaneously benefit individuals as well as firms. \added{An immediate implication of this result is that privacy regulation that limits data sharing can harm all market participants by preventing such mutually beneficial sharing.}

We undertake our study within the context of e-commerce, in which competing firms engage in imperfect competition over a set of consumers. Firms have heterogeneous data about consumers, and sharing all their data with one another leads to a more efficient outcome and benefits consumers. However, the increased competition brought about by data sharing  lowers firm profits, and so firms do not willingly participate. On the other hand, when partial data can be shared, firms can do so in a way that increases their respective profits, but lowers consumer welfare. Most interestingly, we show that a middle ground exists---that partial data can be shared in a way that increases firm profits, while at the same time also increasing market efficiency and benefiting consumers.

We demonstrate our ideas in the most-studied model of imperfect competition, namely, that of \citet{hotelling1929stability}. There are two firms, each located at a different endpoint of a unit interval, with a unit mass of consumers distributed across this interval. In order to study data sharing we depart from the standard model and suppose that the firms may have data about some of the consumers. We model data about a consumer as information about that consumer's location within the interval. Thus, we suppose that there are consumers whose locations are known only to the first firm, consumers whose locations are known only to the second firm, consumers whose locations are not known to either firm, and consumers whose locations are known to both firms.


Within this model, we first show that indiscriminate sharing of data is harmful to the firms. In particular, we compare firms' baseline profits---those attained in equilibrium with no data sharing---with their profits under full data-sharing, in which each firm shares all of its location data with the other firm. Relative to the baseline of no sharing, full data-sharing increases market efficiency and consumer welfare, but lowers firm profits. We then show that if firms share data about only some of the consumers, as opposed to all, then it is possible to attain outcomes that are strictly Pareto-improving---weakly increasing the utilities of all market participants, with at least one strict increase---relative to the baseline. In particular, our main result is the design of a mechanism that  increases firm profits as well as each consumer's welfare, and, in particular, maximizes joint firm profits subject to being Pareto-improving. In addition, we also design a mechanism for which consumer welfare is maximal, subject to being Pareto-improving.

\paragraph{Organization of the paper} Immediately following is a review of the related literature, after which we formally describe the model. Our analysis then proceeds in several stages of increasing generality. We begin in Section~\ref{sec:one-segment} with an analysis of data sharing in the simple case where only one firm has data about consumers, whereas the other has no data at all. Although interesting in their own right, the mechanisms we design here also help to develop intuition and serve as building blocks for our subsequent mechanisms in Sections~\ref{sec:two-segments} and~\ref{sec:four-segments}, in which both firms are assumed to have some data about consumers. Throughout, we make a standard distributional assumption for Hotelling games, namely, that consumers are uniformly distributed within the unit interval. As a robustness check, in Appendix~\ref{sec:non-uniform} we relax this assumption and generalize our main result---the construction of a firm-optimal, Pareto-improving mechanism---to general consumer distributions.

\paragraph{Related literature}
%
Although information sharing between firms has been studied in a variety of settings,\footnote{These include oligopolistic competition \citep{clarke1983collusion,raith1996general}, financial intermediation \citep{pagano1993information,jappelli2002information,gehrig2007information}, supply chain management \citep{ha2008contracting,shamir2016public}, and
competition between data brokers \citep{gu2019data,ichihashi2020competing}.} our paper is most-closely related to that of competitive price discrimination---see, for example, the surveys of \cite{stole2007price} and \cite{fudenberg2012digital}. One of the main insights from this literature is that when firms have more data about consumers, competition between them is more intense, leading to lower prices. And although this is generally beneficial to consumers, it harms firms. An immediate corollary is that, in general, full data-sharing (which leads to firms having more data about consumers) is harmful to firms, echoing the concern quoted above from the \cite{EC2020}.

Two papers that specifically analyze the effects of data sharing within a Hotelling model include \cite{jentzsch2013targeted}, and \cite{clavora2021effects}.
\cite{jentzsch2013targeted} study a model in which each of two firms may have data both about consumers' locations and about their transportation costs, and consider the eight permutations in which each firm may have either a dataset about locations, a dataset about transportation costs, both datasets, or neither datasets. They then analyze the market effects of firms sharing one or both of their (full) datasets with each other, and provide conditions under which sharing is beneficial to the firms. In particular, it is harmful for firms to share both datasets, but may be beneficial for them to share their (full) data on transportation costs alone. They also show that such partial sharing is typically detrimental to consumers.

\cite{clavora2021effects} studies a Hotelling model in which locations are two-dimensional, and firms hold all data about one dimension, both dimensions, or neither dimension. He analyzes the various scenarios in terms of firm profits and consumer welfare, with a particular emphasis on the comparison to the regimes of full privacy (neither firm has any data) and no privacy (both firms have full data). Interestingly, \citeauthor{clavora2021effects} shows that total firm profits are hump-shaped in the amount of information they hold; for example, the scenario in which each firm holds data about a different dimension yields higher profits than both full privacy and no privacy. Unlike our paper, \cite{clavora2021effects} does not focus on the effects of data sharing, nor on the possibility of sharing partial data about a particular dimension. In addition, in his setup, changing the informational allocation typically has an ambiguous effect on consumers. In contrast, we focus on the possibility of designing data-sharing mechanisms in a way that increases all participants' welfare.

In terms of modeling, our paper is most closely related that of \cite{montes2019value}.  \citeauthor{montes2019value} consider a one-dimensional Hotelling model in which consumers' locations may be known by one, both, or neither firm. Their concern is not data sharing between the firms, but rather the optimal strategy of a data broker who sells the data to the firms. They also consider the effects of a consumer-side technology that allows consumers the ability to protect their privacy.

Our paper is also related to dynamic models of price discrimination \citep[e.g.,][]{liu2006customer,kim2010customer,de2017behavior,choe2018pricing,choe2020behavior}, in which the allocation of data across firms is endogenous: in a first period players compete without any data about consumer locations; then, they obtain information about locations of consumers who bought from them, and compete again in a second period.
Within this branch of the literature, the work of \cite{choe2020behavior} is most closely related to our paper. In that paper, the authors study information sharing in the two-period model, where, before the first period, firms decide whether or not they will share all their information with each other between the two periods. A main result is that (full) information sharing is beneficial to the firms, as it softens price competition in the first period, and this benefit is higher than the loss due to greater competition in the second period. At the same time, due to the decrease in first-period competition, information sharing is harmful to consumers.

There are some other papers that are related to ours, although they study data sharing in somewhat different contexts. 
For example, \cite{liu2006customer} consider a model that has both vertical and horizontal differentiation, and in which firms may collect data about consumer loyalty. They show that  firms have an incentive to share data only under horizontal differentiation, and that such sharing harms consumers.
\cite{belleflamme2020competitive} study a model of Bertrand competition in which firms asymmetrically and probabilistically obtain information about consumers' willingness to pay, and are thus imperfectly able to target consumers. They show that, if firms are asymmetric in their ability to target consumers, then partial data sharing can be beneficial to the firms. Their paper differs from ours in its model---Bertrand competition with imperfect targeting rather than imperfect competition---but also in its focus. First, they do not study optimal sharing schemes, as we do. Furthermore, the sharing schemes they do consider are not Pareto-improving: even when they improve firm profits, they necessarily leave some or all consumers worse off. Finally, \cite{argenziano2021data} study a model where where a consumer interacts sequentially with two different firms, and analyze the question of how data linkages between the firms---akin to full data sharing of the upstream firm with the downstream firm---affects consumer welfare. They consider various privacy regulations, and show that some are beneficial to consumers and others harmful.

Another interesting approach is that of  \cite{de2020crowdsourcing}, who study an 
infinitely repeated setting in which firms compete for market share. The authors construct a data-sharing mechanism that increases firm profits, a mechanism that also features partial sharing.  Unlike our approach, however, \citeauthor{de2020crowdsourcing}'s mechanism relies on the repeated nature of the interaction they study, and the authors invoke folk-theorem-type arguments to show that cooperation can be sustained. 

In an attempt to capture data and information effects in a more general setting, \cite{osorio2020model} introduces a linear-demand model in which firms have imperfect information about consumers, and studies conditions under which firms have incentives to engage in full data-sharing. In his framework, if the gains in targeting from better prediction are higher than the costs due to added competition, data sharing is beneficial to the firms. Even when sharing is beneficial to the firms, however, it can be harmful to consumers due to increased prices. 

Another related paper is \cite{gradwohl2020coopetition}, in which we design optimal data-sharing schemes for firms engaged in taste prediction or targeted advertising. In that paper's setting, better prediction is beneficial to both firms (in aggregate) and to consumers. Thus, there is no conflict between firm profits and consumer welfare, the conflict on which we focus in the current paper.

Finally, our work is also related to several papers that consider orthogonal questions. First, the literature on the sale of data by a data provider \citep[e.g.,][and others]{admati1986monopolistic,bergemann2018design,montes2019value} studies how a data provider can maximize profits by selling data to a monopolists or competing firms who use this data to price discriminate. Our paper differs in that information is not sold by a third party to maximize profits, but rather is shared firms with one another, thereby affecting their respective market positions. Second, the work of \cite{ali2020voluntary} considers a setting where the consumers have information about their location, and may choose to share it with one or both firms so as to intensify competition or lower prices.  Our paper studies an orthogonal question, as we assume firms already have some differential information about consumers, and focus on whether they will share it with each other.

\added{Finally, within the vast literature on privacy, the papers most closely related to ours are those of \cite{taylor2004consumer}, \cite{acquisti2005conditioning}, \cite{conitzer2012hide}, \cite{belleflamme2016monopoly}, \cite{montes2019value}, and \cite{chen2020competitive}, all of which examine the effects of privacy regulation on outcomes and welfare in a variety of price-discrimination models. In addition,
\cite{de2020data} study the competitive effect of disallowing all data sharing between firms,
and \cite{gal2019competitive} analyze the anti-competitive effects of the GDPR.}

\section{Model and Preliminaries}\label{sec:model}
We focus on a standard Hotelling model, in which a unit mass of consumers is spread over the unit interval. 
There are two firms: firm $A$ is located at $\theta_A=0$, and
firm $B$ is located at $\theta_B=1$. Each consumer chooses at most one firm from which to purchase a good.
Consumers derive value $v$ from the good, but pay two costs: the price, and a linear transportation cost that scales
with the distance between the consumer and the firm providing the good.
Thus, a consumer located at $\theta$ who buys from firm $i$ at price $p_i$ obtains utility $v-p_i-t\abs{\theta-\theta_i}$, where
$t$ is the marginal transportation cost. 
We assume throughout that the market is covered---namely, that $v> 2t$---so that all consumers
purchase a good even when there is a monopolist firm. Finally, we also assume for simplicity that firms' marginal costs are 0, and
so their profit from the sale of a good is equal to the price. These are all standard assumptions in Hotelling games.

The standard setup consists of a two-stage game: First, firms simultaneously set prices; second, consumers choose a firm
and make a purchase. In the unique subgame perfect equilibrium of this game,\footnote{The equilibrium is unique up to the choice of the indifferent consumer located at $\theta=0.5$.} firms' prices are $p_A=p_B=t$, consumers
in $[0,0.5)$ buy from $A$, and consumers in $(0.5,1]$ buy from $B$ \citep[see, e.g.,][]{belleflamme2015industrial}.

In this paper we will consider a variant of the standard model by supposing that firms may have additional information about
some of the consumers. In particular, we will suppose that, for some consumers, one or both firms know the
location of that consumer on the unit interval. For such consumers, firms will be able to offer a {\em personalized price}---a 
special offer specifically tailored to that consumer. If a firm does not know a consumer's location,
however, then it cannot distinguish between that consumer and all other consumers whose location it does not know. All such consumers are offered the same {\em uniform price} by the firm.

To model how data is spread across the firms we suppose that, instead of a single unit interval of consumers,
there are four segments of consumers: 
segment $S_B$ consists of those consumers whose location only firm $B$ knows, $S_A$ 
consists those consumers whose location only firm $A$ knows, $S_{\emptyset}$ consists those consumers whose location 
neither firm knows, and $S_{AB}$ consists those consumers whose location both firms know.
Each segment consists of a unit interval, with firm $A$ located at 0 and firm $B$ at 1. 
In most of the paper we assume for simplicity that, within
each segment, consumers are uniformly distributed; however, we relax this assumption in Appendix~\ref{sec:non-uniform}. 
Finally, we assume that
firms' higher-order beliefs are as follows: if a firm knows a consumer's location, then it also knows the segment on which
that consumer is located (and so it knows whether or not the other firm knows the consumer's location); if a firm does not
know a consumer's location, then it also does not know the consumer's segment (and so it does not know whether or not the other firm knows the consumer's location).

Given this informational environment, a data-sharing mechanism $M=(M_B, M_A)$ between firms specifies a subset $M_B$ of $S_B$ and a subset $M_A$ of $S_A$, with the interpretation that firm $B$ shares with
firm $A$ the locations of $S_B$ consumers on $M_B$, and firm $A$ shares with firm $B$ the locations of $S_A$
consumers on $M_A$. 
Note that neither firm can share locations of consumers in $S_{\emptyset}$ since neither knows them, and neither firm 
need share locations of consumers in $S_{AB}$ since both
already know them. 
Formally, a mechanism $M$

Two simple examples of data-sharing mechanisms are one that involves {\em no sharing} between firms, $M=(\emptyset,\emptyset)$, and one that involves {\em full sharing}, $M=(S_B,S_A)$.
As the names imply, in the former no firm shares any information with the other, whereas in the latter both firms share
all their information with each other. Alternatively, a firm may share data about a subset of the consumers on its segment. For example, firm $B$ may share with firm $A$ the locations of consumers in $M_B=[x, y]\subseteq [0,1]$ on $S_B$. In this case, if a consumer located in $[x,y]$ on $S_B$ arrives, both firms will know that consumer's location. On the other hand, 
if a consumer located in $[0,1]\setminus [x,y]$ on $S_B$ arrives, firm $B$ will know that consumer's location, and firm $A$ will be able to deduce that the consumer is not located in $[x,y]$ on $S_B$.

One important desideratum of data-sharing mechanisms is that they be {\em individually rational (IR)}: That the expected utility of each firm with data sharing be at least as high as without data sharing. A data-sharing mechanism should be IR if we expect firms to participate.

Our main focus will be on mechanisms that are not only IR, but also {\em Pareto-improving}: that when data sharing takes place, (i) the expected utility of each firm and {\em every} consumer be at least as high as without data sharing, and that (ii) either firm $A$'s profits, firm $B$'s profits, or total consumer welfare be strictly higher.

In our analysis, we consider the following order of events:
\begin{enumerate}
\item Firms engage in a data-sharing mechanism $M=(M_B,M_A)$.
\item Firms simultaneously and publicly announce uniform prices, $p_A$ and $p_B$.
\item A consumer arrives, and all firms who know the consumer's location $\theta$ simultaneously offer 
that consumer a personalized price, $p_A(\theta)$ and $p_B(\theta)$.
\item The consumer chooses a firm from which to buy, and payoffs are realized.
\end{enumerate}

Note that firms share data, and then simultaneously announce their uniform prices, before consumers arrive. 
After a consumer arrives to the market, the firms who know
the consumer's specific location simultaneously offer personalized prices. 
When a firm offers a consumer a personalized price, this offer subsumes the firm's original uniform price.
Thus, a firm's uniform price will apply only to those consumers who will not subsequently be offered a personalized price
by that firm.

Importantly, when firms set personalized prices, they know
the uniform price set by the other firm in the previous stage. 
This is the standard timing considered in the literature \citep[see, e.g.,][]{thisse1988strategic,choudhary2005personalized,choe2018pricing,montes2019value,chen2020competitive}.\footnote{An alternative model that we do not analyze
is one in which firms set uniform and personalized prices simultaneously, for each consumer. \citet{montes2019value} show that, in this case, a (pure) equilibrium may fail to exist.}

For any fixed mechanism $M$, we will be consider the pure subgame perfect equilibrium of the game that starts with data-sharing mechanism $M$. Such an equilibrium always exists. We will be interested in designing mechanisms $M$ that lead to equilibria with high firm profits and high consumer welfare.

Data sharing between the firms has a direct effect and an indirect effect. The direct effect is that a firm
that has obtained information about more consumers' locations via the sharing mechanism can now offer personalized prices to more consumers, affecting both firms' personalized prices offered to those consumers
in equilibrium. The indirect effect is that data sharing may change
the set of consumers for whom the uniform price applies, since additional consumers will now be offered personalized
prices. And since the uniform price is determined in equilibrium in part by the locations of consumers to whom that price will
apply, a change in the set of consumers may effect a change in the equilibrium uniform price.

%
%


\section{Warm Up: One Segment}\label{sec:one-segment}
We begin our analysis with the simplest informational environment in which data sharing has some bite. Namely,
we assume that the entire mass of consumers is located on segment $S_B$: firm $B$ has data and knows the locations of all consumers, whereas firm $A$ has no data about any consumer.
This environment serves to illustrate some of our main insights on the
benefits of data sharing to both firms and consumers. We will also use the mechanisms and intuitions developed here as building blocks for mechanisms in more complex environments.

Because we will use these mechanisms as building blocks, it will  be helpful to weaken the IR requirement for this setting.
In particular, we will say that a mechanism is {\em jointly IR} if the {\em sum} of firm profits under the mechanism is higher than without data sharing. Note that an IR mechanism is jointly IR, but that the reverse may not hold. However, a jointly IR mechanism can always be made IR if monetary transfers between the firms are feasible.

In the subsequent subsections we analyze six different mechanisms:
\begin{enumerate}
\item No data-sharing: this is the baseline.
\item Full data-sharing: while beneficial to the consumers, this mechanism is harmful to firms. It is not jointly IR.
\item Firm $B$ shares data on consumers in $[1/4, 1/2)$: this mechanism is jointly IR, strictly increases consumer welfare,
and leaves total firm profits unchanged. Furthermore, this mechanism maximizes consumer welfare subject to being jointly IR.
\item Firm $B$ shares data on consumers in $[\delta, 1/2)$, where $\delta$ is close to 0: this mechanism is IR and strictly increases both firms' profits, but at a cost to consumer welfare.
\item Firm $B$ shares data on consumers in $[\alpha, 1/2)$, with a particular $\alpha \in (\delta, 1/4)$:
this mechanism is  jointly IR and trades off firm profits and consumer welfare, leading to a strict increase in both.
\item Firm $B$ shares data on consumers in $[1/4, 3/8)$:
this mechanism maximizes the sum of firm profits subject to not harming any consumer.
\end{enumerate}

Before turning to the six mechanisms, we state a preliminary note and a simple lemma.
First, note that, in the single-segment setting, since $B$ has data about every consumer it will always offer a personalized price. 
Thus, $B$'s uniform price applies to no consumers, and so for the remainder of the section we will ignore it.
Second, the simple lemma:
\begin{lemma}\label{lem:baseprice}
Regardless of firm $A$'s uniform price $p_A\geq 0$, in every equilibrium consumers in $(1/2, 1]$ will purchase from firm $B$.
\end{lemma}
\begin{proof}
Fix a consumer $\theta\in(1/2,1]$. If both $A$ and $B$ offer $\theta$ a personalized price, then since $\theta$ is closer
to $B$ the latter will always be able to offer a lower price. Thus, $\theta$ will purchase from $B$. If $A$ does
not offer a personalized price, then consumer $\theta$ chooses between buying from $A$ at uniform price $p_A$ and getting utility
$v-p_A - t\theta$, or buying from $B$ at personalized price $p_B(\theta)$ and getting utility 
$v-p_B(\theta)-t(1-\theta)$. Since $\theta> 1/2$, firm
$B$ can always choose a positive $p_B(\theta)$ such that $v-p_B(\theta)-t(1-\theta)> v-p_A - t\theta$.
\end{proof}

We now proceed to the analysis of the five different mechanisms.

\subsection{No Data-Sharing}
We begin with the case of no data-sharing. In this case, firm $A$ chooses some uniform price $p_A$, whereas
firm $B$ personalizes a price to each consumer, if possible making the latter indifferent between buying from $A$ and from 
$B$.\footnote{Assume that if a consumer is indifferent, he purchases from the firm offering a personalized price.} 
Thus, given $p_A$, firm $B$ charges personalized price
$p_B(\theta) = \max\{0,  p_A-t(1-2\theta)\}$ to a consumer located at $\theta$. Observe that, at these prices,
consumers in $[0, \mu)$ purchase from firm $A$, whereas consumers in $[\mu, 1]$ purchase from $B$, where
$$\mu = \mu(p_A) = \frac{1}{2}-\frac{p_A}{2t}.$$
Given this, firm $A$ maximizes its profit $p_A \cdot \mu(p_A)$
by solving
$$\max_{p_A}p_A\left(\frac{1}{2}-\frac{p_A}{2t}\right).$$
This is maximized at $1/2 - p_A/t  = 0$ and so at
$p_A = t/2$. At this price, firm $B$'s personalized prices are $\max\{t(2\theta-1/2),0\}$, and the first indifferent consumer is 
$\mu(t/2) = 1/4$. Thus, consumers $[0,1/4)$ purchase from $A$, whereas consumers $[1/4, 1]$ purchase from B.
Finally, at these prices firm profits are $\pi_A = t/8$ and 
$$\pi_B = \int_{1/4}^1 t(2\theta-1/2)d\theta = \frac{9t}{16},$$
whereas consumer welfare is
$$CW = \int_0^1 \max\{v-\theta t - p_A, v-t(1-\theta) - p_B(\theta)\} d\theta 
=  \int_0^1 (v - t/2 - \theta t)  d\theta = v-t.$$
This is summarized in the following proposition:
\begin{proposition}
Suppose firm $B$ has all consumers' information, and firm $A$ has none. 
Then firm $A$'s uniform price is $p_A=t/2$, firm $B$'s personalized
price is $p_B(\theta)=\max\{t(2\theta-1/2),0\},$
profits are $\pi_A=t/8$ and $\pi_B=9t/16$, and consumer welfare is $CW = v-t$.
\end{proposition}

\subsection{Full Data-Sharing}
Suppose now that the firms engage in full data-sharing, in which firm $B$ shares all its data with firm $A$.
In this case, then, both firms know the location of every consumer, and so both engage in personalized
pricing. This setting is analyzed by \citet{taylor2014consumer}, who show the following:

\begin{proposition}\label{prop:one-full}
Under full data sharing, personalized prices are
$$p_A(\theta) = \max\{t(1-2\theta),0\}~~~and~~~p_B(\theta)=\max\{t(2\theta-1),0\},$$
profits are $\pi_A=\pi_B=t/4$, and consumer welfare is $CW = v-3t/4$.
\end{proposition}

Notice that under full sharing, consumers are better off than under no sharing. For the firms, naturally firm $B$ is better
 off with no sharing and firm $A$ with sharing. Importantly, however, note that total profits $\pi_A+\pi_B$ are higher under
 no sharing ($11t/16$) than under full sharing ($t/2$). Thus, full sharing is not jointly IR, and there is no price firm $B$ could charge firm $A$ in exchange for fully sharing its data that would make both firms better off.

\subsection{Sharing $[\eps, 1/2]$}
In this section we analyze three data-sharing mechanisms in which firm $B$ shares data about consumers in 
$[\eps, 1/2)$ for various values of $\eps\in(0,1/4]$.
After presenting our general result, we consider the three specific cases that map to mechanisms 3-5 above, 
and show the corresponding properties:
\begin{itemize}
\item[3.] If $\eps=1/4$, then the mechanism strictly increases consumer welfare while leaving total firm profits unchanged;
\item[4.] If $\eps$ is close to 0, then \textbf{both} firms attain higher profits than under no sharing;
\item[5.] There exist values of $\eps$ for which the mechanism strictly increases both consumer welfare and total firm profits.
\end{itemize}
We now state the general result.

\begin{proposition}\label{prop:one-eps}
Consider the mechanism in which firm $B$ shares with firm $A$ the locations of consumers in $[\eps,1/2)$, where $\eps\in(0,1/4]$. Then $A$'s uniform price is $p_A=t(1-2\eps)$,  $A$'s personalized prices for consumers $\theta\in [\eps,1/2)$ are
$p_A(\theta) = t(1-2\theta)$, and  $B$'s personalized prices are
$$p_B(\theta) = \begin{cases} t(2\theta-2\eps)&\mbox{if } \theta\in[1/2,1] \\
t(2\theta-1)&\mbox{if } \theta\in[\eps,1/2)\\
0 & \mbox{otherwise.} \end{cases}$$
Profits are $\pi_A=t(1/4-\eps^2)$ and $\pi_B=t(3/4-\eps)$, and consumer welfare is $CW = v-t(5/4-\eps-\eps^2)$.
\end{proposition}
The proof appears at the end of this section.

We now return to the three mechanisms discussed above. 
\begin{itemize}
\item[3.] If $\eps=1/4$,
then firm profits are $\pi_A=3t/16$ and $\pi_B=t/2$, and so $\pi_A+\pi_B=11t/16$. This is the same as total profits under
no sharing, and thus the mechanism is jointly IR. Consumer welfare under this  mechanism is $v-15t/6$, which is strictly higher than the consumer welfare of $v-t$ under no sharing. In addition, observe that here, firm $A$'s uniform price is $t/2$, which is the same as that firm's uniform price absent data sharing. This implies that all consumers are weakly better off under this data sharing mechanism when compared to no sharing: Consumers located in $[0,1/4)$ buy from $A$ at the same uniform price, consumers in $[1/2,1]$ buy from $B$ at the same personalized price, and the remaining consumers switch from $B$ to $A$ and pay a lower personalized price.  Finally, in this mechanism there is no deadweight loss, since each consumer purchases the good from the closer firm. Thus, because the joint-IR constraint binds, there is no mechanism that can lead to higher consumer welfare while also being jointly IR. 

\item[4.] As $\eps\rightarrow 0$, firm profits are $\pi_A\rightarrow t/4$ and $\pi_B\rightarrow 3t/4$. These profits
are strictly higher for \textbf{both} firms than under no sharing, and so the mechanism is IR. 
However, this comes at the expense of consumer
welfare, which now approaches $v-5t/4$.

\item[5.] Does there exist some $\eps$ for which both total profits $\pi_A+\pi_B$ and consumer welfare are 
strictly higher than under no
sharing?\footnote{One might additionally ask whether there exists an $\eps$ for which welfare increases, and also \textbf{both}
firms' profits increase (and not just total profits). Such an $\eps$ must additionally satisfy
$\pi_A=t \left(\frac{1}{4}-\eps^2\right) > \frac{t}{8}$
and
$\pi_B =t\left(\frac{3}{4}-\eps\right) > \frac{9t}{16}.$
However, it is straightforward to show that no $\eps$ can simultaneously satisfy these and (\ref{eqn:1}).} In order to achieve this, $\eps$ must satisfy
$$\pi_A+\pi_B =t \left(\frac{1}{4}-\eps^2\right) + t\left(\frac{3}{4}-\eps\right) > \frac{t}{8}+\frac{9t}{16}$$
and
\begin{eqnarray}v-t\left(\frac{5}{4} - \eps - \eps^2\right) >v-t.\label{eqn:1}\end{eqnarray}
We can observe that these are both satisfied whenever
$$\frac{1}{4}<\eps+\eps^2 < \frac{5}{16}.$$
\end{itemize}

\begin{proofof}{Proposition~\ref{prop:one-eps}}
We begin with $A$'s uniform price. By Lemma~\ref{lem:baseprice}, that uniform price will apply only to consumers
in $[0, \eps)$: those in $[\eps,1/2)$ will pay $A$'s personalized price, whereas the rest will buy from $B$.
Given this, firm $A$ maximizes its profit $p_A \cdot \mu(p_A)$ on the segment $[0,\eps)$ by 
solving
$$\max_{p_A}p_A\left(\frac{1}{2}-\frac{p_A}{2t}\right)~~\mbox{s.t.}~~\frac{1}{2}-\frac{p_A}{2t}\leq \eps.$$
The optimal solution here is a corner one, with $1/2-p_A/(2t)=\eps$ and so $p_A = t(1-2\eps)$.

For consumers in $[\eps,1/2)$ firm $A$ offers the personalized price $p_A(\theta) = t(1-2\theta)$.
Firm $B$ offers the same personalized prices as with full sharing on these consumers, but offers
price
$p_B(\theta) = t(2\theta-2\eps)$ to consumers in $[1/2,1]$ as their outside option is to buy from $A$ at price
$p_A=t(1-2\eps)$, and $B$'s chosen price leaves them indifferent.

We now calculate profits given these prices. First,
$$\pi_A = \eps\cdot t(1-2\eps) + \int_\eps^{1/2} t(1-2\theta)d\theta = t\left(\frac{1}{4}-\eps^2\right).$$
Next,
$$\pi_B = \int_{1/2}^1 t(2\theta - 2\eps)d\theta = t\left(\frac{3}{4}-\eps\right).$$

Finally, consumer welfare is
\begin{align*}
CW&=\int_0^{\eps}\left(v-t(1-2\eps)-t\theta\right)d\theta+
\int_{\eps}^{1/2}\left(v-t(1-2\theta)-t\theta\right)d\theta+
\int_{1/2}^1\left(v-t\left(2\theta-2\eps\right)-t(1-\theta)\right)d\theta\\
&=v-t\left[\int_0^{\eps}(1-2\eps+\theta)d\theta+
\int_{\eps}^{1/2}(1-\theta)d\theta+
\int_{1/2}^1(1+\theta-2\eps)d\theta\right]\\
&=v-t\left[\eps(1-2\eps)+\frac{\eps^2}{2}+\frac{3}{8}-\eps-\frac{\eps^2}{2}
+\frac{1}{2}\left(\frac{3}{2}-2\eps\right)+\frac{1}{8}\right]\\
&= v- t\left(\frac{5}{4}-\eps-\eps^2\right).
\end{align*}
\end{proofof}

\subsection{Firm-optimal mechanism}\label{sec:one-firm-opt}
Recall that if firm $B$ shares data about consumers in $[\delta,1/2)$, where $\delta$ is close to 0, both firms gain but consumers are harmed.
In this section we design a mechanism that is best for firms subject to leaving every consumer unharmed.
In particular, we design a mechanism that is {\em jointly firm-optimal}---namely, that maximizes
the sum of firms' profits---subject to the condition that every consumer's welfare under the mechanism
is weakly higher than her welfare under no sharing. 

\begin{proposition}\label{prop:one-firm-opt}
Consider the mechanism in which firm $B$ shares with firm $A$ the locations of consumers in $[1/4,3/8)$.
This mechanism is weakly beneficial to every consumer. Furthermore, it is jointly firm-optimal relative to all other mechanisms that are
weakly beneficial to every consumer.
\end{proposition}

\begin{proof}
Observe that, with no sharing, consumers in $[0,1/4)$ purchase from $A$ at uniform price $p_A=t/2$ and obtain utility
$v-t(\theta+1/2)$, whereas consumers
in $[1/4,1]$ purchase from $B$ at personalized price $p_B(\theta)=t(2\theta-1/2)$ and obtain utility
$v-t(2\theta-1/2)-t(1-\theta)=v-t(\theta+1/2)$. 

Now consider some data-sharing mechanism that does not harm any consumer. In order for the consumer's utility not to
decrease after sharing, 
one of the following three conditions must be satisfied:
\begin{enumerate}
\item The consumer purchases from the same firm as with no sharing, but at a (weakly) lower price.
\item The consumer switches to the other firm, and that other firm is closer to the consumer. The consumer may pay a higher price,
but the price increase is no higher than the savings in lower transportation costs.
\item The consumer switches to the other firm, and that other firm is farther from the consumer. The consumer pays a lower price,
and the price decrease is higher than the increase in transportation costs.
\end{enumerate}

No mechanism that maximizes joint firm profits will facilitate condition 3. Thus, in any jointly firm-optimal mechanism,
consumers in $[0,1/4)$ will purchase from $A$ and consumers in $(1/2, 1]$ will purchase from $B$.

Now, since sharing data about a consumer in $[0,1/4)$ will lead to a higher price for that consumer, it will no longer satisfy the
conditions above, and so no data can be shared
about such consumers. In addition, since sharing data about a consumer in $(1/2, 1]$ will lead to a lower price for that consumer,
no jointly firm-optimal mechanism will facilitate sharing about such consumers either. Thus, the only consumers about whom
data may be shared are those in $[1/4, 1/2]$.

Without data sharing, consumers in $[1/4, 1/2]$ purchase from $B$ at personalized price $t(2\theta-1/2)$. If $B$ shares data about
a consumer $\theta \in [1/4, 1/2]$, then in the resulting equilibrium that consumer will purchase from $A$ at personalized price
$t(1-2\theta)$. Note that this leads to consumer utility $v-t\theta - t(1-2\theta) = v-t(1-\theta) \geq v-t(\theta+1/2)$, where the right-hand-side
is the consumer's utility without sharing, and so consumer $\theta$
will be better off with sharing (an instance of condition 2 above). 
Firm profits also change: with no sharing, firm $B$ receives $t(2\theta-1/2)$,
whereas with sharing, firm $A$ receives $t(1-2\theta)$. However, $t(1-2\theta)> t(2\theta-1/2)$ for $\theta \in [1/4, 1/2]$ if and only if $\theta \in [1/4,3/8)$. Thus, the claimed mechanism is jointly firm-optimal.
\end{proof}


%

\section{Two Segments}\label{sec:two-segments}
Suppose now that there are two segments, $S_B$ and $S_A$, that consumers are evenly split between them, and that
a mass of $1/2$ of consumers is uniformly distributed on each segment.
Consider any mechanism $M$ from Section~\ref{sec:one-segment},
and let $M^2$ be the mechanism in which $B$ shares data with $A$ about consumers in $S_B$ as in $M$, and $A$ shares
data with $B$ about consumers in $S_A$ symmetrically as in $M$. What are the properties of $M^2$?

It is straightforward to see that $M^2$ inherits all the properties of $M$. More interestingly, if $M$ is jointly IR, then $M^2$ is IR.
Furthermore:
\begin{itemize}
\item The mechanism in which firm $B$ shares data on $S_B$ consumers in $[1/4, 1/2)$, and firm $A$ shares data on $S_A$ 
consumers in $(1/2, 3/4]$ is {\em IR}, weakly increases every consumer's welfare, strictly increases total consumer welfare, and leaves total firm profits unchanged. Furthermore, this mechanism maximizes consumer welfare subject to being jointly IR. This last statement holds because there is no deadweight loss, and the joint IR constraint binds. Thus,  this mechanism is not only Pareto-improving, but it is also consumer-optimal relative to all Pareto-improving mechanisms.
\item The mechanism in which firm $B$ shares data on $S_B$ consumers in $[1/4, 3/8)$, and firm $A$ shares data on $S_A$ 
consumers in $(5/8, 3/4]$ is IR, and maximizes the sum of firm profits subject to the constraint that every consumer is (weakly) better off
than under no sharing. Thus, this mechanism is not only Pareto-improving, but it is also jointly firm-optimal relative to all Pareto-improving mechanisms.
\end{itemize}

\section{Four Segments}\label{sec:four-segments}

Suppose now that there are four segments---$S_A$, $S_B$, $S_\emptyset$, and $S_{AB}$---and that consumers are evenly split between them. In this case the analysis is slightly more involved, since uniform prices need to take segment $S_{\emptyset}$ into account. In the following we consider four mechanisms: no data-sharing, full data-sharing, the consumer-optimal mechanism, and the firm-optimal mechanism.

\subsection{No Data-Sharing}
We begin with the case of no data-sharing. 
The main difference between the four-segment case and the one- and two-segment cases of the previous sections is that here,
when firm $A$ (resp., $B$) chooses a uniform price, that price no longer applies only to consumers in $S_B$ (resp., $S_A$),
but also to consumers on $S_{\emptyset}$. The firms will thus optimize their prices differently from the other settings, leading also
to different personalized pricing.

Before we begin our analysis, observe that both firms offer consumers on $S_{AB}$ personalized prices. By 
the analysis of \citet{taylor2014consumer} (restated as Proposition~\ref{prop:one-full} above),
these prices are $p_A(\theta) = \max\{t(1-2\theta),0\}$ and $p_B(\theta)=\max\{t(2\theta-1),0\}$.

Now suppose uniform prices $p_A$ and $p_B$ are fixed. Then firm $B$ 
will personalize a price to each consumer on $S_B$, if possible making the latter indifferent between buying from $A$ and from 
$B$. Similarly,  firm $A$ 
will personalize a price to each consumer on $S_A$, also making consumers indifferent between the personalized price and $B$'s
uniform price.
Thus, given $p_A$ and $p_B$, personalized prices are
$p_B(\theta) = \max\{0,  p_A+(2\theta-1)t\}$ and $p_A(\theta)=\max\{0,  p_B+(1-2\theta)t\}$. 
Observe that at these prices,
$S_B$ consumers in $[0, \mu_1)$ purchase from firm $A$, whereas $S_B$ consumers in $[\mu_1, 1]$ purchase from $B$, where
$$\mu_1 = \mu(p_A) = \frac{1}{2}-\frac{p_A}{2t}.$$
Furthermore, at these prices, $S_{\emptyset}$ consumers in $[0, \mu_3)$ purchase from firm $A$, whereas $S_{\emptyset}$ consumers in $[\mu_3, 1]$ purchase from $B$, where
$$\mu_3 = \mu(p_A,p_B) = \frac{1}{2}-\frac{p_A-p_B}{2t}.$$

Given this, firm $A$ maximizes its profit $p_A \cdot (\mu_1+\mu_3)$
by solving
$$\max_{p_A}p_A\left[\left(\frac{1}{2}-\frac{p_A}{2t}\right)+\left(\frac{1}{2}-\frac{p_A-p_B}{2t}\right)\right].$$
The first-order condition is
$$1-\frac{2p_A}{t} + \frac{p_B}{2t} = 0.$$
The symmetric case for firm $B$ has the symmetric first-order condition
$$1-\frac{2p_B}{t} + \frac{p_A}{2t} = 0,$$
and the solution to this system is
$$p_A=p_B = \frac{2t}{3}.$$
Observe that, at these prices, the indifferent consumer on $S_B$ is at $1/6$, the indifferent consumer on $S_A$ is at $5/6$,
and the indifferent consumers on $S_{\emptyset}$ and $S_{AB}$ are at $1/2$.

Finally, at these prices firm profits are 
$$\pi_A = \frac{1}{4}\left[\frac{1}{6}\cdot \frac{2t}{3} + \int_0^{5/6}\left(\frac{2t}{3}+(1-2\theta)t\right)d\theta
+ \frac{1}{2}\cdot \frac{2t}{3} + \int_0^{1/2}(1-2\theta)t d\theta\right],$$
where the additive terms refer to profits from $S_B$, $S_A$, $S_{\emptyset}$, and $S_{AB}$, respectively.
These profits, which, due to symmetry, hold also for firm $B$, are
$$\pi_A=\pi_B=\frac{25t}{72}.$$

To simplify the analysis of consumer welfare, observe that consumers on $S_B$ are indifferent between purchasing
from $A$ at price $p_A=2t/3$ and purchasing from $B$ at personalized price $p_B(\theta)$. Similarly, consumers on $S_A$
are indifferent between purchasing from $B$ at price $p_B=2t/3$ and purchasing from $A$ at personalized price $p_A(\theta)$.
Due to symmetry, conditional on $S_B$ (resp., $S_A$), consumer welfare is thus
$$\int_0^1\left(v-\theta t - \frac{2t}{3}\right)d\theta = v - \frac{7t}{6}.$$
Conditional on $S_{\emptyset}$, consumer welfare is
$$\int_0^{1/2}\left(v-\theta t - \frac{2t}{3}\right)d\theta +\int_{1/2}^1\left(v-(1-\theta) t - \frac{2t}{3}\right)d\theta 
= v - \frac{11t}{12}.$$
Finally, conditional on $S_{AB}$, consumer welfare is
$$\int_0^{1/2}\left(v-\theta t - t(1-2\theta)\right)d\theta +\int_{1/2}^1\left(v-(1-\theta) t - t(2\theta-1)\right)d\theta 
= v - \frac{3t}{4}.$$
Overall, consumer welfare is thus
$$CW = v - \frac{1}{4}\left[\frac{7t}{6} + \frac{7t}{6} + \frac{11t}{12} + \frac{3t}{4} \right]
 = v-t.$$
This is summarized in the following proposition:
\begin{proposition}\label{prop:four-none}
In the four-segment model with no data-sharing, profits are
$\pi_A=\pi_B=25t/72$ and consumer welfare is $CW = v-t$.
\end{proposition}

\subsection{Full Data-Sharing}
Suppose now that the firms engage in full data-sharing, in which firm $B$ shares all its $S_B$ data with firm $A$ and
firm $A$ shares all its $S_A$ data with firm $B$.
Unlike the case of a single segment, it is not immediately clear that full data-sharing harms firms. With a single segment,
full data-sharing leads to more competition over every consumer, thus driving down prices and profits. In contrast, when
there are four segments, there is also an opposing force: Because full data-sharing leads to direct competition on each consumer
on segments $S_B$, $S_A$, and $S_{AB}$, it actually drives up uniform prices, leading to greater profits from consumers on $S_{\emptyset}$.
Nonetheless, as Proposition~\ref{prop:four-full} below states, this positive effect on profits does not suffice, and overall 
firms are worse off after sharing.

With full data-sharing, both firms know the location of every consumer on  $S_B$, $S_A$, and $S_{AB}$, 
and so both engage in personalized
pricing for these consumers. Neither knows the locations of consumers on $S_{\emptyset}$, so uniform prices apply to them. 
The former setting is analyzed by \citet{taylor2014consumer} (restated as Proposition~\ref{prop:one-full} above), who show that profits are 
$\pi_A=\pi_B=t/4$, and consumer welfare is $CW = v-3t/4$. The latter is the standard Hotelling game setting, where profits are 
$\pi_A=\pi_B=t/2$, and consumer welfare is $CW = v-5t/4$. Overall, under full data-sharing
profits are thus
$$\pi_A=\pi_B = \frac{3}{4}\cdot \frac{t}{4} + \frac{1}{4}\cdot\frac{t}{2} = \frac{5t}{16}$$
and consumer welfare is
$$CW = v - \frac{3}{4}\cdot \frac{3t}{4} - \frac{1}{4}\cdot\frac{5t}{4} = v - \frac{7t}{8}.$$
This is summarized in the following proposition.
\begin{proposition}\label{prop:four-full}
Under full data-sharing, personalized prices for consumers in $S_B$, $S_A$, and $S_{AB}$ are
$$p_A(\theta) = \max\{t(1-2\theta),0\}~~~and~~~p_B(\theta)=\max\{t(2\theta-1),0\},$$
 uniform prices (which apply to consumers in $S_{\emptyset}$) are $p_A=p_B = t$,
profits are $\pi_A=\pi_B=5t/16$, and consumer welfare is $CW = v-7t/8$.
\end{proposition}

Observe that, as in the one- and two-segment settings, firms are better off with no sharing than with full sharing, whereas
consumers are better off with full sharing than with no sharing.


\subsection{Consumer-Optimal Mechanism}
In this section we design an IR mechanism that maximizes consumer welfare subject to satisfying the joint IR constraint.
The mechanism is a variant of mechanism 3 from Section~\ref{sec:one-segment}, 
in which firm $B$ shares data on consumers in $[1/4, 1/2)$. Recall that in the one-segment case, when there is no data-sharing, 
firm $A$ sells to consumers in $[0,1/4)$. Thus, sharing data about consumers in $[1/4, 1/2)$ causes them to switch to firm
$A$, but does not affect firm $A$'s uniform price (see Proposition~\ref{prop:one-eps} with $\eps=1/4$).
In the four-segment setting, however, the only $S_B$ consumers who purchase from $A$ are in $[0,1/6)$. Thus,
in our mechanism, firm $B$ will share data about $S_B$ consumers in $[1/6,1/2)$. Symmetrically, firm $A$ will share
data about $S_A$ consumers in $(1/2, 5/6]$.

\begin{proposition}\label{prop:four-firm-consumer-opt}
Consider the mechanism in which firm $B$ shares with firm $A$ the locations of $S_B$ consumers in $[1/6,1/2)$,
and  firm $A$ shares
with firm $B$ the locations of $S_A$ consumers in $(1/2, 5/6]$. In equilibrium, firms' uniform prices are $p_A=p_B=2t/3$.
This mechanism is IR, beneficial to every consumer, and maximizes consumer welfare relative to all jointly IR mechanisms.
\end{proposition}

For some intuition, consider firm $A$'s uniform price, which applies to consumers on $S_B$ and on $S_{\emptyset}$. 
Without data sharing, that price is $p_A=2t/3$, and $S_B$ consumers $[0,1/6)$ buy from $A$. After data sharing, what
 can firm $A$ gain by changing $p_A$? Observe first that lowering $p_A$ cannot be helpful; it does not affect the
 number of consumers on $S_B$ that purchase from $A$---since consumers $[1/6,1/2)$ anyway purchase from $A$ at personalized 
 prices---but only lowers revenue from those consumers. Lowering $p_A$ does increase $A$'s market share of consumers on $S_{\emptyset}$,
 but at lower revenue. However, $p_A=2t/3$ optimally balances this tradeoff, as it does in the no-sharing case.
 
 Similarly, raising $p_A$ cannot be helpful. If $p_B=2t/3$, then without sharing $p_A=2t/3$ is a best response, 
 by Proposition~\ref{prop:four-none}. Any profitable deviation by $A$ to a higher $p_A$ will also be profitable
 under no sharing, contradicting the fact that $2t/3$ is an equilibrium uniform price. Thus, $p_A=p_B=2t/3$ are
  equilibrium uniform prices under this sharing mechanism. In the proof, we show that these uniform prices
  are in fact the unique equilibrium uniform prices.

\begin{proof}
We begin with uniform prices. Recall that, as claimed in Proposition~\ref{prop:four-none}, when there is no data sharing
then $p_A=p_B=2t/3$, and the first indifferent consumer on $S_B$ is at $1/6$.

Fix some arbitrary potential uniform prices $p_A$ and $p_B$. Firm $A$'s uniform price will only potentially apply to
$S_B$ consumers in $[0,1/6)$, since $S_B$ consumers in $[1/6,1/2)$ will pay $A$'s personalized price,
whereas $S_B$ consumers in $[1/2,1]$ will pay $B$'s personalized price. In addition, $A$'s uniform price
will potentially apply to all $S_{\emptyset}$ consumers. Thus, if firm $A$ chooses a uniform price less than $2t/3$, then it can only
increase its consumer base on $S_{\emptyset}$ but not on $S_B$. If $A$ chooses a uniform price greater than $2t/3$, then
it potentially shrinks its consumer base on both $S_B$ and $S_{\emptyset}$.

Overall, firm $A$ maximizes its profit by solving for the larger of
$$\max_{p_A\in[0,2t/3)}p_A\left[\frac{1}{6}+\left(\frac{1}{2}-\frac{p_A-p_B}{2t}\right)\right]$$
or
$$\max_{p_A\in[2t/3,\infty)}p_A\left[\left(\frac{1}{2}-\frac{p_A}{2t}\right)+\left(\frac{1}{2}-\frac{p_A-p_B}{2t}\right)\right].$$
The first maximization has a corner solution at $p_A=2t/3$, and the second has an interior solution
at
$$p_A=\frac{t}{2}+\frac{p_B}{4}.$$

The same solutions apply symmetrically to firm $B$. Regardless of whether $A$'s uniform price is $p_A=2t/3$
or $p_A=t/2+p_B/4$,
plugging $p_A$ into  $B$'s first-order condition
$$p_B=\frac{t}{2}+\frac{p_A}{4}$$
yields the unique solution $p_A=p_B=2t/3$.

These uniform prices, together with the corresponding personalized prices, affect firms' profits relative to no sharing
only on $S_B$ consumers in $[1/6,1/2)$ and $S_A$ consumers in $(1/2,5/6]$. This is because uniform prices are the same,
and so all other consumers face the same prices as with no sharing. How do profits change on these two subsegments?
Consider $S_B$ consumers in $[1/6,1/2)$. With no sharing, each such consumer $\theta$ purchased from $B$
at personalized price $p_B(\theta) = p_A+(2\theta-1)t = t(2\theta-1/3)$. With sharing, each such consumer purchases from
$A$ at personalized price $p_A(\theta)=t(1-2\theta)$. However,
$$\int_{1/6}^{1/2}t(2\theta-1/3)d\theta=\int_{1/6}^{1/2}t(1-2\theta)d\theta,$$
and so revenue lost by $B$ on $S_B$ is exactly recovered by $A$. Symmetrically, revenue lost by $A$ on $S_A$
is exactly recovered by $B$. This implies that both $A$ and $B$ are indifferent between sharing and no sharing,
and so the mechanism is IR.

Finally, the fact that the mechanism maximizes consumer welfare relative to all jointly IR mechanisms follows from the observations that the joint-IR constraint binds, and that deadweight loss is minimal since all consumers purchase from the closer firm.
\end{proof}

\subsection{Firm-Optimal Mechanism}
In this section we design an IR mechanism that maximizes joint firm profits subject to not harming any consumers.
The mechanism is a variant of the mechanism from Section~\ref{sec:one-firm-opt}, 
in which firm $B$ shares data on consumers in $[1/4, 3/8)$. Again, recall that in the one-segment case, when there is no data-sharing, 
firm $A$ sells to consumers in $[0,1/4)$. Thus, sharing data about consumers in $[1/4, 3/8)$ causes them to switch to firm
$A$, but does not affect firm $A$'s uniform price (see Proposition~\ref{prop:one-firm-opt}). It also only causes those consumers 
that increase total firm profits to switch.
In the four-segment setting, however, the only $S_B$ consumers who purchase from $A$ are in $[0,1/6)$. Thus,
in our mechanism, firm $B$ will share data about $S_B$ consumers in $[1/6,1/3)$. Symmetrically, firm $A$ will share
data about $S_A$ consumers in $(1/3, 5/6]$.

\begin{proposition}\label{prop:four-firm-opt}
Consider the mechanism in which firm $B$ shares with firm $A$ the locations of $S_B$ consumers in $[1/6,1/3)$,
and  firm $A$ shares
with firm $B$ the locations of $S_A$ consumers in $(2/3, 5/6]$.
This mechanism is Pareto-improving and jointly firm-optimal relative to all mechanisms that are
weakly beneficial to every consumer.
\end{proposition}

The intuition is similar to that of Proposition~\ref{prop:four-firm-consumer-opt} above.
Without data sharing, $A$'s uniform price is $p_A=2t/3$, and $S_B$ consumers $[0,1/6)$ buy from $A$. After data sharing, 
firm $A$ cannot gain by raising $p_A$, as in Proposition~\ref{prop:four-firm-consumer-opt}. However, unlike 
Proposition~\ref{prop:four-firm-consumer-opt}, here firm $A$ could increase its market share on $S_B$ by lowering prices,
since only $S_B$ consumers $[1/6,1/3)$ buy from $A$ at personalized prices. If $A$ were to lower prices enough, then
potentially also $S_B$ consumers in $[1/3,1/2)$, who currently buy from $B$ at personalized prices, may switch to $A$.
Of course, this comes at considerable cost to $A$, as the uniform price needs to be lower than $t/3$ in order to attract these consumers.
This is unprofitable, and so $p_A=p_B=2t/3$ remain the equilibrium prices.
Furthermore, in the equilibrium here no consumers are harmed, and firm profits are maximized subject to no consumers
being harmed, for the same reason as in the one-segment case of Proposition~\ref{prop:four-firm-consumer-opt}.

\begin{proof}
We begin with uniform prices. Recall that, as claimed in Proposition~\ref{prop:four-none}, when there is no data sharing
then $p_A=p_B=2t/3$, and the first indifferent consumer on $S_B$ is at $1/6$.

Firm $A$'s uniform price will only potentially apply to
$S_B$ consumers in $[0,1/6)\cup [1/3,1/2)$, since $S_B$ consumers in $[1/6,1/3)$ will pay $A$'s personalized price,
whereas $S_B$ consumers in $[1/2,1]$ will pay $B$'s personalized price. In addition, $A$'s uniform price
will potentially apply to all $S_{\emptyset}$ consumers. If firm $A$ chooses a uniform price less than $2t/3$ but greater than $t/3$, then it can only
increase its consumer base on $S_{\emptyset}$ but not on $S_B$. If firm $A$ chooses a uniform price less than $t/3$, then it can 
increase its consumer base on $S_{\emptyset}$ and also $S_B$, getting some of the consumers in $[1/3,1/2)$. If $A$ chooses a uniform price greater than $2t/3$, then it potentially shrinks its consumer base on both $S_B$ and $S_{\emptyset}$.

Overall, firm $A$ maximizes its profit by solving for the largest of
$$\max_{p_A\in[0,t/3)}p_A\left[\left(\frac{1}{2}-\frac{p_A}{2t}-\frac{1}{6}\right)+\left(\frac{1}{2}-\frac{p_A-p_B}{2t}\right)\right],$$
$$\max_{p_A\in[t/3,2t/3)}p_A\left[\frac{1}{6}+\left(\frac{1}{2}-\frac{p_A-p_B}{2t}\right)\right],$$
or
$$\max_{p_A\in[2t/3,\infty)}p_A\left[\left(\frac{1}{2}-\frac{p_A}{2t}\right)+\left(\frac{1}{2}-\frac{p_A-p_B}{2t}\right)\right].$$
The first maximization has a corner solution at $p_A=t/3$,
the  second maximization has a corner solution at $p_A=2t/3$, and the third has an interior solution
at
$$p_A=\frac{t}{2}+\frac{p_B}{4}.$$

The same solutions apply symmetrically to firm $B$, leading to a unique equilibrium with $p_A=p_B=2t/3$.

These uniform prices, together with the corresponding personalized prices, affect firms' profits relative to no sharing
only on $S_B$ consumers in $[1/6,1/3)$ and $S_A$ consumers in $(1/2,2/3]$. This is because uniform prices are the same,
and so all other consumers face the same prices as with no sharing. How do profits change on these two subsegments?

Consider some consumer $\theta\in [1/6, 1/2)$ on $S_B$. If such a consumer were to face personalized pricing by both $A$ and $B$,
then he would choose $A$ and pay personalized price 
$p_A(\theta)=t(1-2\theta)$. On the other hand, if $A$ does not know the consumer's location, then that consumer faces $A$'s uniform price $p_A=2t/3$ and $B$'s personalized price $p_B(\theta)=t(2\theta-1/3)$, and chooses the latter. Now, observe that the
set $[1/6,1/3)$ is precisely the subset of consumers in $[1/6,1/2)$ for whom total firm profits $p_A(\theta)=t(1-2\theta)$ are higher
than $p_B(\theta)=t(2\theta-1/3)$. Thus, the mechanism under consideration maximizes total firm profits subject to
only affecting $S_B$ consumers in $[1/6, 1/2)$ and, symmetrically, $S_A$ consumers in $(1/2,5/6]$.

Next, as in Proposition~\ref{prop:one-firm-opt}, the only other way for firms to increase profits, while not harming any consumer, 
is to cause other consumers to switch
to the closer firm and charge a higher price. However, all other consumers are already purchasing from the closest firm. Thus,
this mechanism maximizes total firm profits subject to not harming any consumer.
\end{proof}

\section{Conclusion}\label{sec:conclusion}
In this paper we analyzed a Hotelling  model of imperfect competition, and showed that  data sharing need not be a zero sum endeavor---benefiting consumers at the expense of firms, or vice versa. In contrast, we designed mechanisms for partial data-sharing that are Pareto-improving, simultaneously benefiting individuals as well as firms. In our analysis we utilized standard assumptions,
such as a uniform distribution of consumer locations. In Appendix~\ref{sec:non-uniform} we show that our results are robust to more general distributions.

Our results have some  implications for regulation. In particular, since data sharing can be beneficial to all market participants, privacy regulation that limits data sharing can harm market participants by preventing such mutually beneficial 
sharing.\footnote{See, for example, the analysis of  \cite{gal2019competitive} on the effects of data-sharing limitations within the EU's General Data Protection Regulation (GDPR).} Of course, privacy regulation is multi-faceted, and its effects on data sharing should be traded off against its many benefits. Regulators thus face the challenging task of avoiding this potential harm and instead steering firms' data sharing in a Pareto-improving direction.


\bibliographystyle{ims}
\bibliography{hotellingDS}
\newpage
\appendix
\begin{center}
\begin{Large}
\textbf{Appendix}
\end{Large}
\end{center}

\section{Robustness: Non-Uniform Consumer Distributions}\label{sec:non-uniform}
The uniformity assumption makes computations tractable; however, as noted by \cite{belleflamme2015industrial}, it often does 
not perform well empirically.
Several papers study Hotelling games with non-uniform distributions \citep[including][]{shilony1981hotelling,tabuchi1995asymmetric,anderson1997location,benassi2008elasticity,azar2015linear}, focusing on conditions under which equilibria exist. As a robustness check on the mechanisms presented in the body of the paper, in this section we study data sharing under such non-uniform distributions.

The setting we study includes arbitrary proportions of consumers on each of the segments, and an arbitrary
distribution of consumers within each segment. As our concern is orthogonal to the issue of equilibrium existence, our results in this section are of the following form. Suppose that when
there is no sharing then there is a pure-strategy equilibrium that yields profits $\pi_A$ and $\pi_B$, and consumer
welfare $CW$. Then there is a Pareto-improving mechanism and jointly firm-optimal relative to all Pareto-improving mechanisms. 
  

More formally, the general model is as follows. There are four segments, labeled $S_i$ for $i\in\{A,B,\emptyset,AB\}$, as in
the uniform model. The prior probability that a consumer is in segment $S_i$ is $q_i$, where $\sum_i q_i = 1$.
Furthermore, the distribution of consumers on each segment $S_i$ has pdf $f_i$ and cdf $F_i$, where $F_i(1)=q_i$.
Note that this general model can capture all our prior settings by considering pdfs that are uniform on each segment,
and additionally setting $q_A=q_{\emptyset}=q_{AB}=0$ for the one-segment model and $q_{\emptyset}=q_{AB}=0$ for the two-segment model.

Our result below, Proposition~\ref{prop:general-simple-case}, considers mild conditions on the distributions of consumers---namely, that they are sufficiently ``balanced''.
Under these conditions, we construct an IR mechanism for data sharing that is strictly Pareto-improving relative to no sharing
and that maximizes joint firm profits subject to not harming any consumer. Following the proof of the result, we provide
two simple settings in which the conditions are satisfied. 


\Rnote{A weaker statement that we can make (also above) is that it maximizes joint firm profits relative to all Pareto-improving
mechanisms.}

\begin{proposition}\label{prop:general-simple-case}
Suppose that with no sharing the equilibrium prices are $p_A$ and $p_B$. Furthermore, suppose that
$$\int_{\alpha_1}^{\frac{1+2\alpha_1}{4}} t(1-2\theta) f_1(\theta)d\theta \geq \int_{\frac{1+2\alpha_2}{4}}^{\alpha_2} \left[p_B+t(1-2\theta)\right] f_2(\theta)d\theta$$
and
$$\int_{\frac{1+2\alpha_2}{4}}^{\alpha_2} t(2\theta-1) f_2(\theta)d\theta \geq \int_{\alpha_1}^{\frac{1+2\alpha_1}{4}}\left[p_A+t(2\theta-1)\right]f_1(\theta)d\theta,$$
where 
$$\alpha_1=\frac{1}{2}-\frac{p_A}{2t}~~~~~\mbox{and}~~~~\alpha_2=\frac{1}{2}+\frac{p_B}{2t}.$$
Then the mechanism in which firm $B$ shares with firm $A$ the locations of $S_B$ consumers in $[\alpha_1,1/4+\alpha_1/2)$,
and  firm $A$ shares
with firm $B$ the locations of $S_A$ consumers in $(1/4+\alpha_2/2, \alpha_2]$ is IR, Pareto-improving, and jointly firm-optimal 
relative to all to other mechanisms that are weakly beneficial to every consumer.
\end{proposition}
\begin{proof}
Consider some consumer $\theta\in [\alpha_1, 1/4+\alpha_1/2)$ on $S_B$. If such a consumer were to face personalized pricing by both firms $A$ and $B$,
then he would choose $A$ and pay personalized price 
$p_A(\theta)=t(1-2\theta)$. On the other hand, if $A$ does not know the consumer's location, then that consumer faces $A$'s uniform price $p_A$ and $B$'s personalized price $p_B(\theta)=p_A+t(2\theta-1)$, and chooses the latter. Now, observe that the
set $[\alpha_1, 1/4+\alpha_1/2)$ is precisely the subset of consumers in $[\alpha_1, 1/2)$ for whom total firm profits $p_A(\theta)=t(1-2\theta)$ are higher
than $p_B(\theta)=p_A+t(2\theta-1)$. Thus, the mechanism under consideration maximizes total firm profits subject to
only affecting $S_B$ consumers in $[\alpha_1, 1/2)$ and, symmetrically, $S_A$ consumers in $(1/2,\alpha_2]$.

Next, as in Propositions~\ref{prop:one-firm-opt} and~\ref{prop:four-firm-opt}, the only other way for firms to increase profits, while not harming any consumer, 
is to cause other consumers to switch
to the closer firm and charge a higher price. However, all other consumers are already purchasing from the closest firm. Thus,
this mechanism maximizes total firm profits subject to not harming any consumer.

Finally, under these mechanisms, the first inequality in the statement of the proposition states that
firm $A$'s profit from $S_B$ consumers relative to no sharing is higher than the firm's loss on $S_A$ consumers.
The second inequality states the same for firm $B$. Thus, if the inequalities are satisfied, then the mechanism
is IR.
\end{proof}

A simple case in which the inequalities in Proposition~\ref{prop:general-simple-case} hold is when the distributions and equilibrium are symmetric:
namely, when  $f_A(\theta)=f_B(1-\theta)$ and $f_{AB}(\theta)=f_{AB}(1-\theta)$ for all $\theta$, and the equilibrium
uniform prices satisfy $p_A=p_B$. 
\Rnote{When the distributions are symmetric then there always exists a symmetric equilibrium,
but there may also exist asymmetric equilibria. Need reference.}

Another case is when consumers are uniformly distributed within each segment, but where consumers may not be evenly
divided between the segments, as long as $q_A$ and $q_B$ are within a factor of 3 of one another. This is formalized in the
following proposition.
\begin{proposition}\label{cor:uniform-non-equal}
Suppose consumers are uniformly distributed within each segment, and that $q_B>0$ and $q_A>0$.
Then the conditions in Proposition~\ref{prop:general-simple-case} are satisfied for every $q_{\emptyset}$ and $q_{AB}$ if and only if $3q_B\geq q_A\geq q_B/3$.
\end{proposition}
\begin{proof}
We begin with uniform prices under no sharing. 
Firm $A$'s uniform price will apply to
$S_B$ consumers in $[0,\mu(p_A))$ and $S_{\emptyset}$ consumers in $[0,\mu(p_A,p_B))$.
Thus, firm $A$ maximizes its profit by solving
$$\max_{p_A}p_A\left[q_B\left(\frac{1}{2}-\frac{p_A}{2t}\right)+q_{\emptyset}\left(\frac{1}{2}-\frac{p_A-p_B}{2t}\right)\right],$$
which has interior solution
$$p_A=\frac{t}{2}+\frac{q_{\emptyset} p_B}{2(q_B+q_{\emptyset})}.$$
The same solution applies symmetrically to firm $B$, but replacing $q_B$ with $q_A$.
Then, plugging the solution to $B$'s  first-order condition
into $A$'s yields the unique solution
$$p_A=\frac{t(2q_B+3q_{\emptyset})}{4(q_B+q_{\emptyset})-\frac{q_{\emptyset}^2}{q_A+q_{\emptyset}}}=\frac{t\left(2q_Bq_A+2q_Bq_{\emptyset}+3q_Aq_{\emptyset}+3q_{\emptyset}^2\right)}{4q_Bq_A+4q_Bq_{\emptyset}+4q_Aq_{\emptyset} + 3q_{\emptyset}^2}$$
and
$$p_B=\frac{t(2q_A+3q_{\emptyset})}{4(q_A+q_{\emptyset})-\frac{q_{\emptyset}^2}{q_B+q_{\emptyset}}}=\frac{t\left(2q_Bq_A+2q_Aq_{\emptyset}+3q_Bq_{\emptyset}+3q_{\emptyset}^2\right)}{4q_Bq_A+4q_Bq_{\emptyset}+4q_Aq_{\emptyset} + 3q_{\emptyset}^2}.$$

These prices yield the values $\alpha_1=\mu(p_A) = 1/2-p_A/(2t)$ and $\alpha_2=\mu(p_B) = p_B/(2t)-1/2$. Consider the first
inequality in Proposition~\ref{prop:general-simple-case}, but with $f_1(\theta)=q_B$ and $f_2(\theta)=q_A$:
$$q_B\int_{\alpha_1}^{\frac{1+2\alpha_1}{4}} t(1-2\theta)d\theta \geq q_A\int_{\frac{1+2\alpha_2}{4}}^{\alpha_2} \left[p_B+t(1-2\theta)\right] d\theta.$$
Simplifying this yields
left-hand-side value
$$q_B\left(\frac{\alpha_1+1/2}{2}-\alpha_1\right)\frac{3p_A}{4} = \frac{3q_B p_A^2}{16t}$$
and right-hand-side value $\frac{q_A p_B^2}{16t}$.
Thus, the first inequality in Proposition~\ref{prop:general-simple-case} holds if and only if $3q_B p_A \geq q_A p_B$.
Since $p_A$ and $p_B$ have the same denominator, the inequality holds if and only if
$$3q_B\left(2q_Bq_A+2q_Bq_{\emptyset}+3q_Aq_{\emptyset}+3q_{\emptyset}^2\right) \geq q_A\left(2q_Bq_A+2q_Aq_{\emptyset}+3q_Bq_{\emptyset}+3q_{\emptyset}^2\right),$$
which holds if and only if
$$3q_Bq_Aq_{\emptyset} + 3q_B\left(2q_Bq_A+2q_Bq_{\emptyset}+2q_Aq_{\emptyset}+3q_{\emptyset}^2\right) \geq q_Bq_Aq_{\emptyset} + q_A\left(2q_Bq_A+2q_Aq_{\emptyset}+2q_Bq_{\emptyset}+3q_{\emptyset}^2\right).$$
A sufficient condition for this last inequality to hold, for any value of $q_{\emptyset}$, is that $3q_B \geq q_A$. Furthermore, 
if $q_{\emptyset}=0$ then this condition is also necessary.

Finally, a symmetric analysis for the second inequality in Proposition~\ref{prop:general-simple-case} yields the necessary and
sufficient condition $3q_A\geq q_B$.
\end{proof}

\end{document}